\documentclass[journal,10pt,draftclsnofoot,onecolumn]{IEEEtran}
\usepackage{times}
\usepackage{stackrel}
\usepackage{subfigure}
\usepackage{epsf}
\usepackage{amsmath,amssymb}
\usepackage{graphicx}
\usepackage{color}
\usepackage{cite}
\usepackage{multirow,tabularx}
\usepackage{ifthen}
\usepackage{times}
\usepackage{tabularx}
\usepackage{epstopdf}
\input{epsf.sty}

\newcommand{\hide}[1]{\ifthenelse{\boolean{false}}{#1}{}}

\newtheorem{lemma}{{\bf Lemma}}

\newtheorem{proposition}{\bf Proposition}

\newtheorem{remark}{{\bf Remark}}

\newcommand{\qed}{\nobreak \ifvmode \relax \else
      \ifdim\lastskip<1.5em \hskip-\lastskip
      \hskip1.5em plus0em minus0.5em \fi \nobreak
      \vrule height0.75em width0.5em depth0.25em\fi}

\newcommand{\barr}{\begin{array}}
\newcommand{\earr}{\end{array}}

\newcommand{\benum}{\begin{enumerate}}
\newcommand{\eenum}{\end{enumerate}}

\newcommand{\bit}{\begin{itemize}}
\newcommand{\eit}{\end{itemize}}

\newcommand{\bc}{\begin{center}}
\newcommand{\ec}{\end{center}}

\newcommand{\bdes}{\begin{description}}
\newcommand{\edes}{\end{description}}

\newcommand{\bfig}{\begin{figure}}
\newcommand{\efig}{\end{figure}}

\newcommand{\bemq}{\begin{quote} \begin{em}}
\newcommand{\eemq}{\end{em} \end{quote}}

\newcommand{\bmp}{\begin{minipage}}
\newcommand{\emp}{\end{minipage}}

\newcommand{\eqn}[1]{(\ref{#1})}

\newcommand{\brac}[1]{\left({#1}\right)}
\newcommand{\sbrac}[1]{\left[{#1}\right]}
\newcommand{\cbrac}[1]{\left\{{#1}\right\}}

\newcommand{\bsp}{\begin{slide*}}
\newcommand{\esp}{\end{slide*}}
\newcommand{\bsl}{\begin{slide}}
\newcommand{\esl}{\end{slide}}

\newcommand{\UN}{\bar{\mathcal{U}}^{(j)}_N}
\newcommand{\Ub}{\bar{\mathcal{U}}}
\newcommand{\U}{{\mathcal{U}}}
\newcommand{\norm}[1]{\|{#1}\|}
\newcommand{\abs}[1]{\left \vert {#1} \right \vert}
\newcommand{\mb}[1]{\mathbb{#1}}
\newcommand{\mf}[1]{\mathbf{#1}}
\newcommand{\bs}[1]{\boldsymbol{#1}}
\newcommand{\mv}[3]{{#1}^{(#3)}_{#2}}

\begin{document}
%
\title{Analysis of Randomized Join-The-Shortest-Queue (JSQ) Schemes in Large Heterogeneous Processor Sharing Systems}

\author{Arpan Mukhopadhyay and Ravi R. Mazumdar,~\IEEEmembership{Fellow,~IEEE}

\thanks{The authors are with the department of Electrical
and Computer Engineering, University of Waterloo, Waterloo, 
ON N2L 3G1, Canada.}
\thanks{e-mail: arpan.mukhopadhyay@uwaterloo.ca, mazum@ece.uwaterloo.ca}}

\maketitle

\begin{abstract}

In this paper we investigate the stability and performance of randomized  dynamic routing schemes for jobs based on the Join-the-Shortest Queue (JSQ) criterion to a 
heterogeneous system of many parallel servers.
In particular we consider servers that use processor sharing but with different server rates and jobs are routed to the server with the smallest occupancy 
among a finite number of randomly selected servers. We focus on the case of two servers that is often referred to as a Power-of-Two scheme. 
We first show that in the heterogeneous setting there can be a loss in the stability region over the homogeneous case
and thus such randomized schemes need not outperform static randomized schemes in terms of mean delay in opposition to the homogeneous case of equal server speeds where the stability region is maximal and coincides with that of static randomized routing. 
We explicitly characterize the stationary distributions of the server occupancies and show that the tail distribution of the server occupancy has a super-exponential behavior as in the homogeneous case as the number of servers goes to infinity.
To overcome the stability issue, we show that it is possible to combine static state-independent scheme with 
randomized JSQ scheme that allows us to recover the maximal stability region combined with the benefits 
of JSQ and such a scheme is preferable in terms of average delay. 
The techniques are based on a mean field analysis where we show that the stationary 
distributions coincide with those obtained under asymptotic independence of the servers and 
moreover the stationary distributions are insensitive to the job size distribution.
\end{abstract}


%
\IEEEpeerreviewmaketitle

\section{Introduction}

A central problem in a multi-server
resource sharing system is to decide
which server an incoming job will be assigned to in order to obtain optimum performance, typically the low average response time.
The problem becomes more challenging when
the number of servers in the system is large
and the servers have different service rates.
Such systems are frequently encountered 
in large web server farms that
accommodate a large number of front end servers of various service rates to
process incoming job requests~\cite{Lu_Performance_2011}.
The load balancing scheme used
plays a key role
in determining the mean sojourn time of jobs
in such systems. Since web applications
such as online search, social networking are extremely delay sensitive,
a small increase in the average response time of jobs may cause
significant loss of revenue and users~\cite{Delay_study_google}. 
Therefore, the main objective of efficient load balancing
is to reduce the mean sojourn
times of jobs in the system. Another desirable property of a routing scheme should be its {\em robustness} 
to heterogeneity of job sizes of which a typical statistical behavior is insensitivity to job size distributions.

The join-the-shortest-queue (JSQ) scheme, commonly used
in small web server farms~\cite{Lu_Performance_2011,Salchow_load_2007}, 
assigns a new arrival to the server having the least number
of unfinished jobs in the system. 
Recently, Gupta {\em et al}~\cite{Gupta_Performance_2007} 
showed that for a system of identical processor sharing (PS) servers
JSQ is nearly optimal in terms of minimizing the mean sojourn time of jobs
and results in near insensitivity of the system to the type of job length
distribution. 


However, a major drawback of the
JSQ scheme, when applied to a system consisting of a large number
of servers, is that it requires the
state information of all the servers in the system to make job assignment
decisions. The use of dynamic randomized algorithms is 
one way to avoid requiring information about all server occupancies.
It has been shown that randomized load balancing schemes based on sampling only a 
few servers provides most of the reduction in mean sojourn times 
associated with JSQ  \cite{Vvedenskaya_inftran_1996}.
Indeed as argued in \cite{Mitzenmacher_thesis,
Mitzenmacher_IEEE_2001,Vvedenskaya_inftran_1996},  most of the gains in average sojourn time reduction are obtained when selecting 2 servers at random referred to as the Power-of-Two rule.
This is also referred to as the SQ(2) scheme.

The literature has treated the SQ(2)
scheme  for a system of identical FCFS servers
assuming exponential job length distribution. The exact analysis is still difficult 
because of the coupling between the servers. However, as the work in~\cite{Graham, Vvedenskaya_inftran_1996, Mitzenmacher_IEEE_2001} 
has shown, when the number of servers goes to infinity any finite collections of servers 
can be viewed as independent. This is termed as {\em asymptotic independence} or {\em propagation
of chaos}. With this insight, 
it was shown that in the large system limit
the stationary tail distribution of the number of remaining jobs
at each server decays doubly exponentially as compared to the
exponential decay in case of the optimal state independent scheme
in which job assignments are done independently of the states of the servers.
Consequently, the SQ(2) scheme results in an exponential reduction in the mean sojourn
time of jobs as compared to the optimal state independent scheme.

The analysis of the SQ(2) scheme was further generalized to
include general job length distributions in~\cite{Bramson_asymp_indep,Bramson_randomized_load_balancing}.
However, the analyses in~\cite{Mitzenmacher_thesis,
Mitzenmacher_IEEE_2001,Vvedenskaya_inftran_1996,
Bramson_asymp_indep,Bramson_randomized_load_balancing}
were restricted to the homogeneous case where the servers in the system are identical
in terms of the server speed.

In this paper, we first  analyze the performance of the classical SQ(2) scheme
with uniform sampling for a large system of PS servers with heterogeneous service rates for which there are no available results. 
The first issue that arises is the issue of the stability region for such systems or the maximum
load that the system can support and still yield finite average sojourn times. In particular, 
we show that the stability region is strictly smaller than the maximal stability
region obtained by restricting the normalized arrival rate below the average
capacity of the system. Thus,  
it is possible that the average sojourn time behavior can be worse than 
static randomized routing schemes. We then provide a detailed analysis of the heterogeneous 
system and provide a complete characterization of the stationary distribution when it exists. 
We show that under the SQ(2) scheme the system is asymptotically insensitive
to the type of job length distributions. 
To overcome the reduction in stability we show that a scheme that combines 
static randomized routing to a server class, i.e., sampling with a bias and SQ(2) with uniform sampling within 
servers of the same class, allows us to recover the maximal stability region. We show that this hybrid scheme 
 retains the gains of the SQ(2) scheme. This scheme is therefore always superior 
in the sense of smaller average sojourn time over static randomized routing schemes.

The techniques are based on a mean field approach that extends 
the methodology used in ~\cite{Vvedenskaya_inftran_1996, Martin_AAP_1999} 
for FCFS queues with  exponential job lengths to the heterogeneous PS scenario. 
We show uniqueness of the solution under stability.  Furthermore in the asymptotic limit the stationary distribution of the server occupancies also coincides 
with that obtained by assuming asymptotic independence for any finite subset of the servers \cite{Graham,Bramson_asymp_indep}.

The organization of the paper is as follows. In Section~\ref{sec:model}, 
we present the system model and provide a description of the load balancing schemes 
studied in this paper. Section~\ref{sec:analysis} presents detailed analyses of 
the load balancing schemes. In Section~\ref{sec:numerics}, numerical results are presented
to compare the different schemes and validate the theoretical analyses. The paper is finally concluded
in Section~\ref{sec:conclusions}. 

\section{System Model}
\label{sec:model}

We consider a system consisting of $N$ parallel 
processor sharing (PS) servers with heterogeneous service rates or capacities.
The service rate, $C$, of a server is defined as the
time rate at which it processes a single job
assigned to it. If $x(t)$ jobs are
present at a server of capacity $C$ at time $t$, then the rate
at which each job is processed at time $t$ is given by $C/x(t)$. 
We assume that a server
can have one of the $M$ possible values of service rate
from the set $\mathcal{C}=\cbrac{C_1,C_2,\ldots,C_M}$.
Define the index set $\mathcal{J}=\cbrac{1,2,\ldots,M}$.
For each $j \in \cal{J}$, let the proportion of servers
with service rate $C_j$ be denoted by $\gamma_j$ ($0 \leq \gamma_j \leq 1$).
Clearly, $\sum_{j=1}^{M} \gamma_j=1$.

Jobs arrive according to a Poisson process 
with rate $N\lambda$. Each job requires a random  
amount of work and the job sizes are independent and identically 
distributed, with a finite mean $\frac{1}{\mu}$. 
The inter-arrival times and the job lengths are assumed
to be independent of each other. 
Upon arrival, a job is assigned to one of the $N$ servers 
according to a randomized load balancing scheme. We now discuss the load balancing
schemes considered in this paper.

\subsection{Scheme~1: Optimal state independent scheme}

As a baseline, we consider a scheme that assigns  an incoming job to a
server with a fixed probability, independent of the 
current state of the servers in the system. We denote by $p_j$, $j \in \cal{J}$, 
the probability that an arrival is assigned to one
of the servers having capacity $C_j$. The probabilities
$p_j$, $j \in \cal{J}$, are chosen in such a way that
the mean sojourn time of the jobs is minimized.
Clearly, in this scheme, no communication is required between the
job dispatcher and the servers as the job assignment decisions 
are made independently of the state of the servers. 

\subsection{Scheme~2: The SQ(2) scheme}

In this scheme, a subset of two servers is selected
from the set of $N$ servers uniformly at random at each arrival instant. 
The arriving job is then 
assigned to the server having the least
number of unfinished jobs among the two chosen servers.
In case of a tie, the job is assigned to any one
of the two servers with equal probability $\frac{1}{2}$.
\footnote{The analysis of the SQ(2) scheme can be readily generalized to the SQ($d$) scheme
where an incoming job is assigned to the least loaded server among $d$ randomly chosen ones
at the cost of more notation and complication. However, 
since the SQ(2) scheme provides most of the improvements, we do not pursue it in this paper.}

\subsection{Scheme~3: A hybrid SQ(2) scheme}

This scheme combines the state independent scheme with the SQ(2) scheme.
In this scheme, upon arrival of a new job,
the router first chooses a capacity value
$C_j$, $j \in \cal{J}$, with probability $p_j$.
Two servers having the selected value of capacity
are then chosen uniformly at random from set of available
servers with having that capacity. The job is then
routed to the server having the least number of unfinished
jobs among the two chosen servers. Ties are broken by tossing
a fair coin. The probabilities $p_j$,
for $j \in \cal{J}$, are chosen in such a way that the mean sojourn time
of jobs in the system is minimized. 

\section{Analysis}
\label{sec:analysis}

In this section, we present the analysis of the load balancing schemes
described in the previous section. Since Scheme~1
is a special case of the more general class of load balancing 
schemes analyzed in~\cite{Altman_Opt_load_2008},
we only state the main analytical results for Scheme~1 in Sec.~\ref{sec:scheme1}
without giving the proofs. These results are used later
to compare the different load balancing schemes considered in this paper. 
The detailed analyses of the SQ(2) scheme and the hybrid SQ(2) scheme
are provided in Section~\ref{sec:scheme2} and Section~\ref{sec:scheme3}, respectively. 

\subsection{Scheme~1: Optimal state independent scheme}
\label{sec:scheme1}

In Scheme~1, a job
is assigned to a server with a fixed probability, independent of
the instantaneous states of the servers in the system. 
Hence, under this scheme, 
the system reduces to
a set of independent parallel M/G/1 processor sharing servers. 
It follows directly from Proposition~1
of~\cite{Altman_Opt_load_2008}, that
there exists probabilities $p_j$, $j \in \cal{J}$,
for which the system is stable under Scheme~1 
if and only if the following condition holds:

\begin{equation}
\lambda \in \Lambda=\cbrac{0 \leq  \lambda < \mu \sum_{j=1}^{M} \gamma_j C_j}.
\label{eq:cond_scheme1}
\end{equation}

%
It was also shown in~\cite{Altman_Opt_load_2008} that,
under the above stability condition, the routing probabilities
$p_j$, $j \in \cal{J}$ can be chosen such that the mean sojourn
time of jobs in the system is minimized. The mean sojourn time
minimization problem, formulated as a convex optimization problem,
was solved in Theorem~1 of~\cite{Altman_Opt_load_2008}. 
It was found that the index set $\mathcal{J}_{\text{opt}}=\cbrac{1,2,\ldots,j^*} \subseteq \cal{J}$ of server capacities
and the loads $\bs{\rho}^*=\cbrac{\rho_1^*, \rho_2^*,\ldots,\rho_M^*}$
in the optimal state independent scheme are given by

\begin{equation}
j^*=\sup\cbrac{j \in \mathcal{J}: \frac{1}{\sqrt{C_j}} < \frac{\sum_{i=1}^{j}
\gamma_i \sqrt{C_i}}{\sum_{i=1}^{j} \gamma_i C_i -\frac{\lambda}{\mu}}}.
\label{eq:optset_scheme1}
\end{equation}

\begin{equation}
\rho_i^*=\begin{cases} 1-\sqrt{\frac{1}{C_i}} \frac{\sum_{k=1}^{j^*}\gamma_k C_k-\frac{\lambda}{\mu}}{\sum_{k=1}^{j^*}\gamma_k \sqrt{C_k}}, &\mbox{if } i \in \cal{J}_{\text{opt}}\\
           0, &\mbox{otherwise.}\end{cases}
\label{eq:optrho_scheme1}
\end{equation}

Here, we have assumed that the the server capacities are ordered as $C_1 \geq C_2 \geq \ldots \geq C_M$.
The optimal routing probabilities $p_j^*$, $j \in \cal{J}$,
can be computed from~\eqn{eq:optrho_scheme1}
by using the relations $\rho_j^*=\frac{p_j^* \lambda}{\gamma_j \mu C_j}$.

\subsection{Scheme~2: The SQ(2) scheme}
\label{sec:scheme2}

In the SQ(2) scheme, the job assignments are done
based on the instantaneous states of two randomly selected servers
in the system. Therefore, unlike the state independent scheme,
in this scheme, the arrival processes to the individual servers
are {\em not} independent of each other. This makes the exact analytical computation
of the stationary distribution very difficult for a finite value of $N$.
However, the mean field approach outlined in~\cite{Vvedenskaya_inftran_1996, Martin_AAP_1999}
or the propagation of chaos arguments used in~\cite{Graham, Bramson_randomized_load_balancing, Bramson_asymp_indep}
allow us to analytically characterize the behaviour of the system under this scheme
in the limit as $N \rightarrow \infty$. It will be later shown through simulation
results that such asymptotic analysis accurately captures the behaviour of
a large but finite system of servers. 

In this paper we say that a Markov process is stable if it is positive Harris recurrent. We now characterize the stability region of the system described in Scheme 2.

Let $N^*$ denote the
smallest positive integer ($> 2$) such that $\gamma_j N^*$ 
is a positive integer for all $j \in \cal{J}$. 
Now, let $\Lambda_k$, for $k \in \mb{N}$, denote the stability
region of the system under Scheme~2 when there are $N=kN^*$ servers in the system.
The following proposition characterizes the sets $\Lambda_k$ for $k \in \mb{N}$.

\begin{proposition}
\label{thm:stability}
For  $\Lambda_k$, $k \in \mb{N}$  defined as above, and $\Lambda$ as given in~\eqn{eq:cond_scheme1},
we have

\begin{equation}
\Lambda \supseteq \Lambda_1 \supseteq \Lambda_2 \supseteq \ldots
\label{eq:set_seq}
\end{equation}
Furthermore, if $\Lambda_{\infty}=\cap_{k=1}^{\infty} \Lambda_k$, then $\Lambda_{\infty}$ is given by

\begin{equation}
\Lambda_{\infty}= \cbrac{0  \leq \lambda < \mu \min_{\mathcal{I}\subseteq \mathcal{J}} \cbrac{\frac{\brac{\sum_{j \in \mathcal{I}} \gamma_j C_j}}{\brac{\sum_{j \in \mathcal{I}} \gamma_j}^2}}}
\label{eq:lam_infty}
\end{equation}
\end{proposition}
 
\begin{proof}
The proof is given in Appendix~\ref{proof:thm_stability}.
\end{proof}

\begin{remark}
From~\eqref{eq:set_seq}, it is clear that
for any finite value of  $N$, the stability region under
Scheme~2 is a subset of that under Scheme~1.  
Further, the stability region under Scheme~2 decreases as $N$ increases keeping the proportions
$\gamma_j$, $j \in \cal{J}$, fixed.
 Hence $ \Lambda_{\infty}$ denotes the region where the system is stable for all $N$. We then show that
in this region the mean field has a unique, globally asymptotically
stable equilibrium point in the space of empirical tail measures that are summable.
\end{remark}

\begin{remark}
\label{rmk:eq_stable}
Under the notation $\nu_j=\lambda/\mu C_j$, it is easy to 
see that
$\lambda < \mu \min_{\mathcal{I}\subseteq \mathcal{J}} \cbrac{\frac{\brac{\sum_{j \in \mathcal{I}} \gamma_j C_j}}{\brac{\sum_{j \in \mathcal{I}} \gamma_j}^2}}$
%
in~\eqn{eq:lam_infty} can be equivalently expressed as 

\begin{equation}
\label{eq:cond_scheme2_2}
\sum_{j \in \mathcal{I}} \frac{\gamma_j}{\nu_j} > \brac{\sum_{j \in \mathcal{I}} \gamma_j}^2 \text{ for all } \mathcal{I}\subseteq \mathcal{J}.
\end{equation}
\end{remark}

\subsubsection{Mean Field Analysis}
Assuming exponential job length distribution (with mean $1/\mu$),
we now characterize the stationary distribution of the system
under the SQ(2) scheme as $N \rightarrow \infty$. To do so we extend the mean field approach of~\cite{Vvedenskaya_inftran_1996, Martin_AAP_1999} 
from the homogeneous scenario to the heterogeneous scenario.

Let $\mathbf{x}_N(t)=\cbrac{x^{(j)}_n(t), 1 \leq j \leq M, n \in \mathbb{Z}_{+}}$
denote the state of the system at time $t$, 
where $x^{(j)}_n(t)=\frac{1}{N \gamma_j} \sum_{n'\geq n} \\ y^{(j)}_{n'}(t)$
and $y^{(j)}_n(t)$ is the number of servers having capacity $C_j$
with exactly $n$ unfinished jobs. Hence, $x^{(j)}_n(t)$
denotes the fraction of servers having capacity $C_j$ with at least
$n$ unfinished jobs. From the Poisson arrival and exponential job size assumptions, for any $N$,
the process $\mathbf{x}_N(t)$ is a Markov process.
The state space of the process $\mf{x}_N(t)$ is given by $\prod_{j=1}^{M}\bar{\cal{U}}^{(j)}_N$, 
where $\bar{\cal{U}}^{(j)}_N$ is defined as
follows:

\begin{equation}
\UN=\{g=(g_n, n \in \mathbb{Z}_{+}): g_0=1, 
g_n \geq g_{n+1} \geq 0, N \gamma_j g_n \in \mathbb{N} \text{ }\forall n \in \mathbb{Z}_{+}\}.
\label{eq:UN}
\end{equation}
We generalize the space $\bar{\cal{U}}^{(j)}_N$ to the space $\Ub$
by removing the last constraint in its definition~\eqn{eq:UN}.
Hence, the space $\Ub$ is defined as follows:

\begin{equation}
\Ub=\{g=(g_n, n \in \mathbb{Z}_{+}): g_0=1, 
g_n \geq g_{n+1} \geq 0 \text{ }\forall n \in \mathbb{Z}_{+}\}.
\label{eq:Ub}
\end{equation}
This space will be required to study the limiting properties of the process $\mf{x}_N(t)$
as $N \rightarrow \infty$. 

We  seek to show the weak convergence of the process
$\mf{x}_N(t)$ as $N \rightarrow \infty$ to the deterministic process 
$\mathbf{u}(t)=\cbrac{u^{(j)}_n(t), n \in \mathbb{Z}_{+},  j \in \cal{J}}$,
governed by the following system of
differential equations that represents the mean field:

\begin{align}
\mathbf{u}(0) &= \mathbf{g},\label{eq:diff1}\\
\dot{\mathbf{u}}(t) &= \mathbf{h}(\mathbf{u}(t)),
\label{eq:diff2}
\end{align} 
where $\mathbf{g} \in \Ub^M$, $\mf{h}(\mf{u})=\cbrac{h^{(j)}_n(\mf{u}), n \in \mathbb{Z}_{+},  j \in \cal{J}}$, 
and for $ j \in \cal{J}$

\begin{align}
h^{(j)}_0(\mathbf{u}) &= 0,\label{eq:diff3}\\
h^{(j)}_n(\mathbf{u}) &= \lambda \brac{u^{(j)}_{n-1}-u^{(j)}_n} \sum_{i=1}^{M} \gamma_i \brac{u^{(i)}_{n-1}+u^{(i)}_n}
-\mu C_j \brac{u^{(j)}_n-u^{(j)}_{n+1}}
\label{eq:diff4}
\end{align}
for all $n \geq 1$.

More specifically, we prove that
if the distribution of $\mf{x}_N(0)$ converges to the Dirac 
measure concentrated at the point $\mf{g} \in \Ub^M$ as $N \rightarrow \infty$,
then the process $\mf{x}_N$ converges weakly
to the deterministic process $\mf{u}$ given by the solution of~\eqn{eq:diff1}-\eqn{eq:diff4}. 
We further show that
under condition~\eqn{eq:cond_scheme2_2}, the system~\eqn{eq:diff1}-\eqn{eq:diff4}
has a unique, globally asymptotically stable
equilibrium point $\mf{P}=\cbrac{\mv{P}{k}{j}, k \in \mathbb{Z}_{+}, 1 \leq j \leq M}$,
obtained by solving the equation $\dot{\mathbf{u}}(t) = \mathbf{h}(\mathbf{u}(t))=0$,
in the space of empirical measures having finite mean.

In the following proposition, we summarize some important properties of the equilibrium point $\mf{P}$
of the system~\eqn{eq:diff1}-\eqn{eq:diff4}.

\begin{proposition}
\label{thm:tail_het_properties}
If there exists a solution $\mf{P}$ of
the equation $\mf{h}(\mf{P})=0$ such that
for each $j \in \cal{J}$, $\mv{P}{0}{j}=1$ and $\mv{P}{k}{j} \downarrow 0$ as $k \rightarrow \infty$, then

i) for each $k \in \mb{Z}_+$ and $j \in \cal{J}$ , 
\begin{equation}
P_{k+1}^{(j)} = \nu_j  \left[\gamma_j \brac{P_k^{(j)}}^2 + P_k^{(j)} \brac{\sum_{\substack{i=1 \\ i \neq j}}^{M} \gamma_i P_k^{(i)}}\right.
\left.+\sum_{\substack{i=1 \\ i \neq j}}^{M} \sum_{l=k}^{\infty} \gamma_i \brac{P_{l+1}^{(i)}P_l^{(j)}-P_l^{(i)}P_{l+1}^{(j)}}\right]. \label{eq:tail_het}
\end{equation}

ii) for each $k \in \mb{Z}_+$ and $j \in \cal{J}$

\begin{equation}
\sum_{j=1}^{M} \frac{\gamma_j}{\nu_j}\mv{P}{k+1}{j}= \brac{\sum_{j=1}^{M}\gamma_j \mv{P}{k}{j}}^2
\label{eq:nice}
\end{equation}

iii) the sequence $\cbrac{\mv{P}{k}{j}, k \in \mb{Z}_+}$ decreases doubly exponentially.
\end{proposition}

\begin{proof}
The proof is given in Appendix~\ref{proof:tail_het_properties}.
\end{proof}

\begin{remark}
A real sequence $\cbrac{z_n}_{n\geq 1}$ is said to decrease doubly
exponentially if and only if there exist positive
constants $L$, $\omega < 1$, $\theta > 1$, and $\kappa$ such that
$z_n \leq \kappa \omega^{\theta^{n}}$ for all $n \geq L$.
Hence, by definition, if a sequence $\cbrac{z_n}_{n\geq 1}$ decays doubly exponentially,
then it is summable, i.e., $\sum_{n=1}^{\infty} z_n < \infty$.
Hence, in view of Proposition~\ref{thm:tail_het_properties}.iii), 
if there exists a solution $\mf{P}$ of
the equation $\mf{h}(\mf{P})=0$ satisfying the hypothesis of Proposition~\ref{thm:tail_het_properties},
then it must be summable.
\end{remark}

Before proving the weak convergence of the Markov process
$\mf{x}_N(t)$ to the deterministic process $\mf{u}(t)$
defined by the systems~\eqn{eq:diff1}-\eqn{eq:diff4},
we need to show that the system  indeed has a unique solution in $\Ub^M$
and there exists a unique equilibrium point $\mf{P}$ of it satisfying
$\sum_{k=1}^{\infty} \mv{P}{k}{j} < \infty$ for each $j \in \cal{J}$.
To do so, it is convenient to define the following space of tail distributions on $\mb{Z}_+$
that has finite first moment.  

\begin{equation}
\U=\{g=(g_n, n \in \mathbb{Z}_{+}): g_0=1, 
g_n \geq g_{n+1} \geq 0 \text{ }\forall n \in \mathbb{Z}_{+}, \sum_{n=0}^{\infty} g_n < \infty\}.
\label{eq:U}
\end{equation}
and the following norm on the spaces $\prod_{j=1}^{M}\UN$, 
$\Ub^M$, and $\U^M$: 

\begin{equation}
\norm{u}=\sup_{1 \leq j \leq M} \sup_{n \in \mathbb{Z}_{+}} \frac{\abs{u^{(j)}_n}}{n+1}.
\label{eq:norm}
\end{equation} 
Note that the space $\Ub^M$ is complete and compact
under the above norm. Henceforth, this norm is understood when we refer to 
convergence or continuity in these spaces. The following proposition
guarantees the existence and uniqueness of solution of the system~\eqn{eq:diff1}-\eqn{eq:diff4}
and its equilibrium point $\mf{P}$. To emphasize the dependence
of the solution of the system~\eqn{eq:diff1}-\eqn{eq:diff4} on the
initial point $\mf{g}$, we shall, at times, denote the solution $\mf{u}(t)$ 
by $\mf{u}(t,\mf{g})$.

\begin{proposition}
\label{thm:diff_eqn}
i) The system~\eqn{eq:diff1}-\eqn{eq:diff4} has a unique solution, $\mf{u}(t, \mf{g})$, 
for all $t \geq 0$, in $\Ub^M$ if $\mf{g} \in \Ub^M$.

ii) Under condition~\eqn{eq:cond_scheme2_2}, there exists a unique equilibrium
point or fixed point $\mf{P}$ of the system~\eqn{eq:diff1}-\eqn{eq:diff4} in the space $\U^M$.
Therefore, $\mf{P}$ satisfies the properties stated in Proposition~\ref{thm:tail_het_properties}.

iii) Under condition~\eqn{eq:cond_scheme2_2},

\begin{equation}
\lim_{t \rightarrow \infty} \mf{u}(t, \mf{g})=\mf{P} \text{ for all } \mf{g} \in \U^M.
\end{equation}
%
%
\end{proposition}

\begin{proof}
The proof is given in Appendix~\ref{proof:diff_eqn}.
\end{proof}

Having established the existence and uniqueness of the
equilibrium point of the system~\eqn{eq:diff1}-\eqn{eq:diff4},
we now proceed to establish the weak convergence as $N \rightarrow \infty$ 
of the process $\mf{x}_N(t)$ to the process $\mf{u}(t,\mf{g})$. This is done
by showing that the generator of the process $\mf{x}_N(t)$
converges to the generator of the deterministic 
map $\mf{g} \mapsto \mf{u}(t,\mf{g})$ as $N \rightarrow \infty$~\cite{Ethier_Kurtz_book}. 

For the Markov process $\mf{x}_N(t)$, the generator
$\mathbf{A}_N$ acting on functions 
$f:\prod_{j=1}^{M}\UN \rightarrow \mathbb{R}$ is
defined as $\mathbf{A}_N f(\mathbf{g})=\sum_{\mf{h} \neq \mf{g}} q_{\mf{g}\mf{h}} \brac{f(\mf{h})-f(\mf{g})}$, where $q_{\mf{g}\mf{h}}$, 
with $\mf{g},\mf{h} \in  \prod_{j=1}^{M} \UN$, denotes the transition
rate from state $\mf{g}$ to state $\mf{h}$.

\begin{lemma}
\label{thm:generator} 
Let $\mathbf{g} \in \prod_{j=1}^{M}\UN$
and $\mathbf{e}(n,j)=\brac{e_k^{(i)}}_{k \in \mb{Z}_+, i \in \cal{J}}$
with $e^{(j)}_n=1$ and $e^{(i)}_k=0$
for all $i \neq j$, $k \neq n$. 
The generator
$\mathbf{A}_N$ of the Markov process $\mf{x}_N(t)$
acting on functions $f:\prod_{j=1}^{M}\UN \rightarrow \mathbb{R}$ is
given by

\begin{multline}
\mathbf{A}_N f(\mathbf{g})= \lambda N \sum_{n \geq 1} \sum_{j=1}^{M} \sum_{i=1}^{M} \gamma_i \gamma_j \sbrac{g^{(j)}_{n-1}-g^{(j)}_{n}}
\times \sbrac{g^{(i)}_{n-1}+g^{(i)}_{n}} \sbrac{f(\mathbf{g}+\frac{\mathbf{e}(n,j)}{N \gamma_j})-f(\mathbf{g})} \\
+ \mu N \sum_{n \geq 1} \sum_{j=1}^{M} \gamma_j C_j \sbrac{g^{(j)}_n-g^{(j)}_{n+1}}\sbrac{f(\mathbf{g}-\frac{\mathbf{e}(n,j)}{N \gamma_j})-f(\mathbf{g})}.\label{eq:gen}   
\end{multline}
\end{lemma}

\begin{proof}
The proof follows by noting that the transition rate from the state $\mf{g}$
to the state $\mf{g}-\mf{e}(n,j)/N \gamma_j$, where $n \geq 1$, is given by
$\mu C_j N \gamma_j \sbrac{g^{(j)}(n)-g^{(j)}(n+1)}$. Similarly,
the transition rate from state $\mf{g}$ to the state
$\mf{g}+\mf{e}(n,j)/N \gamma_j$, where $n \geq 1$, is given by
$\lambda N \sbrac{g^{(j)}_{n-1}-g^{(j)}_{n}}\sum_{i=1}^{M}
\gamma_i \gamma_j\sbrac{g^{(i)}_{n-1}+g^{(i)}_{n}}$.
\end{proof}

For $t \geq 0$, the transition semigroup operator 
$\mathbf{T}_N(t)$ generated by the operator
$\mathbf{A}_N$ and acting on functions 
$f:\prod_{j=1}^{M}\UN \rightarrow \mathbb{R}$ is defined by
$\mathbf{T}_N(t) f=\exp\brac{t \mathbf{A}_N }f$.
The following proposition establishes the convergence
of the semigroup $\mf{T}_N(t)$ to the semigroup of the
of the map $\mf{g} \mapsto \mf{u}(t,\mf{g})$.

\begin{proposition}
\label{thm:convergence_semigroup}
For any continuous function $f:\Ub^M \rightarrow \mb{R}$ and $t \geq 0$,

\begin{equation}
\lim_{N \rightarrow \infty} \sup_{\mf{g} \in \Ub^M} \abs{\mf{T}_N(t) f(\mf{g}) - f(\mf{u}(t,\mf{g}))}=0
\end{equation}
and the convergence is uniform in $t$ within any bounded interval.
\end{proposition}

\begin{proof}
The proof follows from the smoothness assumptions  on $f:\Ub^M \rightarrow \mb{R}$.
We omit the technical details.
\end{proof}

From Theorem 2.11 of Chapter 4 of~\cite{Ethier_Kurtz_book},
the above proposition implies that $\mf{x}_N \Rightarrow \mf{u}$ as $N \to \infty$, where
$\Rightarrow$ denotes weak convergence. This implies the weaker result
that $\mf{x}_N(t) \Rightarrow \mf{u}(t)$ for each $t \geq 0$. It also implies
that any limit point of the sequence of invariant measures $\cbrac{\pi_N}$ 
of the processes $\cbrac{\mf{x}_N}$ is an invariant measure of the map $\mf{g} \mapsto \mf{u}(t,\mf{g})$.
We now show that, under condition~\eqn{eq:cond_scheme2_2}, there
is at most one such limit point which is given by the Dirac
measure concentrated at the equilibrium point $\mf{P}\in \U^M$ of 
the system~\eqn{eq:diff1}-\eqn{eq:diff4}.

\begin{proposition}
\label{thm:final_convergence}
Under the condition~\eqn{eq:cond_scheme2_2}, the Markov process $\mf{x}_N(t)$
is positive recurrent for all $N$ and hence has a unique invariant distribution
$\pi_N$ for each $N$. Moreover, $\pi_N \rightarrow \delta_P$ weakly as 
$N \rightarrow \infty$, where $\delta_{\mf{P}}$ is as defined in Proposition~\ref{thm:diff_eqn}, i.e.,

\begin{equation}
\lim_{N \rightarrow \infty} \mb{E}_{\pi_N} f(\mf{g})=f(\mf{P})
\end{equation}
for all continuous functions $f:\Ub^M \rightarrow \mb{R}$.
\end{proposition}

\begin{proof}
The first part of the theorem is a direct consequence of Remark~\ref{rmk:eq_stable} 
following Proposition~\ref{thm:stability}. The weak convergence of the stationary
distributions $\pi_N$ to $\delta_{\mf{P}}$ follows by the arguments
 in Theorem~4.(ii) of~\cite{Martin_AAP_1999} {\em mutatis mutandis}.
\end{proof}

\begin{remark}
The above results establish that the following interchange holds:
\begin{equation}
\label{eq:interchange}
\lim_{t\to \infty} \lim_{N\to \infty} {\mf{x}_N(t)}= \lim_{N\to\infty} \lim_{t\to \infty} {\mf{x}_N(t)}=\mf{P},
\end{equation}
where the limits are in the sense of weak convergence. 

Due to
exchangeability of states among servers having the same capacity, the
above interchange of limits also implies that in the limit as $N \to \infty$
the servers in the system evolve independently of each other \cite{Graham}. More specifically, the tail
distribution of number of pending jobs at time time $t\geq 0$ at a server
of capacity $C_j$ in the limiting system is given by $\cbrac{\mv{u}{n}{j}(t,\mf{g}), n\geq 0}$,
independent of any other server in the system, 
where $\cbrac{\mv{g}{k}{i}, k \geq 0}$, for $i \in \cal{J}$, is the initial tail distribution
of server occupancies at any type $i$ server. Further, the stationary tail distribution
of server occupancies for a type $j$ server in the limiting system
is also independent of all other servers and is given by $\cbrac{\mv{P}{n}{j}, n\geq 0}$.
This property is formally known as {\em propagation of chaos}. 
\end{remark}

\subsubsection{Insensitivity} So far, we have assumed that the job lengths are 
i.i.d exponential random variables. We now show that the 
stationary distribution of the mean field coincides with the stationary
distribution obtained when the queues are independent at equilibrium.
This will imply that 
stationary distribution of server occupancies in the limiting system is insensitive to the job 
length distribution and only depends on their means.

 
\begin{proposition}
\label{thm:dist_het}
Assume that condition~\eqn{eq:cond_scheme2_2} and  
asymptotic independence of queues in equilibrium in the mean field limit, i.e.,
 for any finite set $B$ of servers, 
\begin{equation}
\Pi^{(B)}=\bigotimes_{n \in B} \pi^{(n)}
\label{eq:asymp_indep}
\end{equation}
where $\pi^{(n)}$ and $\Pi^{(B)}$ denote
the marginals of $\Pi$ for the $n^{\textrm{th}}$ server
and for the servers in set $B$, respectively.

Then,  in equilibrium, the arrival process of jobs
at any given server in the limiting system
becomes a state dependent Poisson process whose intensity is given by:
\begin{equation}
\lambda_k =\lambda \sum_{i=1}^{M} \gamma_i \brac{P_k^{(i)}+P_{k+1}^{(i)}},
\label{eq:lambdak_het}
\end{equation} 
where $P_{k}^{(j)}$, for $j \in \cal{J}$
and $k \in \mathbb{Z}_+$, denotes the stationary probability
that a server with capacity $C_j$ has at least $k$ unfinished
jobs. Moreover, we have $P_0^{(j)}=1$, for all $j \in \cal{J}$
and $P_k^{(j)}$, for $k \in \mathbb{Z}_+$ and $j \in \cal{J}$, satisfy~\eqn{eq:tail_het}.

Furthermore, the stationary distribution of a server depends on the job size only through its mean, or the queues are {\em insensitive}.
\end{proposition}

\begin{proof}
Consider any particular server (say server 1) in the system.
Consider the arrivals that have server 1 
as one of its two possible 
destinations. These arrivals constitute the {\em potential arrival 
process} at the server. The probability 
that the server is selected as a potential destination server for
a new arrival is $\brac{1-\frac{\binom{N-1}{2}}{\binom{N}{2}}}=\frac{2}{N}$. 
Thus, from Poisson thinning,
the potential arrival process to a server is a Poisson process with rate 
$\frac{2}{N}\times N\lambda=2\lambda$.

Next, we consider the arrivals that actually join server 1. 
These arrivals constitute the actual arrival process 
at the server. For finite $N$, this process is 
not Poisson since a potential arrival to server 1 actually 
joins server 1 depending on the number of users present at 
the other possible destination server. However, as 
$N \rightarrow \infty$, due to the asymptotic independence 
property stated in~\eqn{eq:asymp_indep}, the numbers of jobs 
present at these two servers become independent of each 
other. As a result, in equilibrium the actual arrival process 
converges to a state dependent Poisson process as $N \rightarrow \infty$.

Consider the potential arrivals at a server 
when the number of users present at the server is $k$. 
This arrival actually joins the server either with 
probability $\frac{1}{2}$ or with probability $1$ 
depending on whether 
the number unfinished jobs at the 
other possible destination server is 
exactly $k$ or greater than $k$, respectively. 
Since a server having capacity $C_j$ is chosen with 
probability $\gamma_j$, 
the total probability that the 
potential arrival joins the server at state $k$ 
is $\sum_{j=1}^{M} \gamma_j \brac{0.5\brac{P_k^{(j)}-P_{k+1}^{(j)}}+P_{k+1}^{(j)}}=0.5\sum_{j=1}^{M} \gamma_j \brac{\mv{P}{k}{j}+\mv{P}{k+1}{j}}$. Therefore,
the rate at which arrivals occur at stake $k$ is given by
$2\lambda \times 0.5\sum_{j=1}^{M} \gamma_j \brac{\mv{P}{k}{j}+\mv{P}{k+1}{j}}$.
This simplifies to~\eqn{eq:lambdak_het}.

Since processor sharing is a symmetric service discipline,
it follows from Theorems 3.10 and 3.14 of~\cite{Kelly_book}
and also from Theorem~4.2 of~\cite{Gupta_Performance_2007}
that the detailed balance equations hold for state dependent Poisson arrivals.
Therefore, we have for $k \in \mb{Z}_+$ and $j \in \cal{J}$ that

\begin{equation}
P_{k+1}^{(j)}-P_{k+2}^{(j)}=\frac{\lambda_k}{\mu C_j} (P_k^{(j)}-P_{k+1}^{(j)}).
\label{eq:temp_het}
\end{equation}
Substituting the value of $\lambda_k$ from~\eqn{eq:lambdak_het} 
into~\eqn{eq:temp_het} and upon further simplification we get~\eqn{eq:tail_het}.
\end{proof}

\begin{remark}
\label{rmk:prop_of_chaos}
Thus, we have shown that under the assumption of asymptotic 
independence of the servers in equilibrium, the 
stationary distribution of server occupancies 
coincides with that of the mean field. 
From the uniqueness of the solution we can 
conclude that asymptotic independence should hold
also for heterogeneous systems.  A direct proof of the asymptotic independence 
is, however, extremely difficult.
The proof remains an open problem even for homogeneous systems and 
any local service discipline~\cite{Bramson_asymp_indep}.
In recent work \cite{Ramanan} propagation of chaos has been established for FCFS systems under general service time distributions.
\end{remark}

\begin{remark}
\label{rmk:rel_freq}
The long run probability that a user joins a server with 
capacity $C_j$ is given by $\frac{N\gamma_j \bar{\lambda}^{(j)}}{N \lambda}$,
where, $\bar{\lambda}^{(j)}= \sum_{k=0}^{\infty} \lambda_k \brac{P_k^{(j)}-P_{k+1}^{(j)}}$ denotes the average arrival rate
to a server having capacity $C_j$.
From~\eqn{eq:lambdak_het} and~\eqn{eq:tail_het},
we obtain that $\frac{\gamma_j \bar{\lambda}^{(j)}}{\lambda}=\frac{\gamma_j P_1^{(j)}}{\nu_j}$ for each
$j \in \cal{J}$. Thus,
the long run probability that a user
joins a server with capacity $C_j$
is $ \frac{\gamma_j P_1^{(j)}}{\nu_j}$.
\end{remark}


\begin{proposition}
The mean sojourn time, $\bar{T}$,
of a job in the heterogeneous system under Scheme~2
is given by

\begin{equation}
\bar{T}=\frac{1}{\lambda}\sum_{j=1}^{M} \sum_{k=1}^{\infty} \gamma_j P_k^{(j)},
\label{eq:soj_time}
\end{equation}
where $P_k^{(j)}$, $k \in \mb{Z}_+$ and $j \in \cal{J}$, are as given in
Proposition~\ref{thm:dist_het}.
\end{proposition}

\begin{proof}
Let $\bar{T}_j$ denote the mean sojourn time
of a user given that it has joined a server having capacity
$C_j$. Now, the expected number of users at a server
having capacity $C_j$ is given by $\sum_{k=1}^{\infty} P_k^{(j)}$.
Let the average arrival rate at the server
be denoted by $\bar{\lambda}^{(j)}$. Thus, applying Little's formula
we have $\bar{T}_j= \frac{\sum_{k=1}^{\infty} P_k^{(j)}}{\bar{\lambda}^{(j)}}$

%
As discussed in Remark~\ref{rmk:rel_freq}, 
the long run probability that a user joins a server
having capacity $C_j$ is $\frac{\gamma_j\bar{\lambda}^{(j)}}{\lambda}$.
Therefore, the overall mean sojourn time is given by
$\bar{T}=\sum_{j=1}^{M} \frac{\gamma_j\bar{\lambda}^{(j)}}{\lambda} \bar{T}_j=\frac{1}{\lambda}\sum_{j=1}^{M} \sum_{k=1}^{\infty} \gamma_j P_k^{(j)}$.
\end{proof}

\subsection{Scheme~3: The hybrid SQ(2) scheme}
\label{sec:scheme3}

We saw that the classical SQ(2) scheme can have smaller stability region than Scheme~1.
We now show that it is possible to recover the 
stability region of Scheme~1
by using the hybrid SQ(2) scheme.

In the hybrid SQ(2) scheme, for each $j \in \cal{J}$,
a service rate $C_j \in \cal{C}$ is selected for a new arrival
with a probability $p_j$. Hence, the aggregate Poisson arrival rate
to the set of $N \gamma_j$ servers, each having capacity $C_j$, is
$p_j N\lambda$. The system may, therefore, be viewed as being composed
of $M$ parallel homogeneous subsystems each working under
the classical SQ(2) scheme. The the $j^{\textrm{th}}$ ($j \in \cal{J}$)
subsystem has $N \gamma_j$ servers of capacity $C_j$ and the total
input rate at this subsystem is $p_j N \lambda$.
Define $\rho_j=\frac{p_j  \lambda}{\gamma_j \mu C_j}$.
From the results of~\cite{Mitzenmacher_thesis, Vvedenskaya_inftran_1996, Bramson_randomized_load_balancing, Bramson_stability_AAP},
we know that the system is stable if and only if $\rho_j < 1$
for all $j \in \cal{J}$. The necessary and sufficient
condition which guarantees the existence of 
routing probabilities $p_j$, $j \in \cal{J}$ for which the system
is stable is given by the following proposition.

\begin{proposition}
There exists probabilities $p_j$, $j \in \cal{J}$, for
which the system is stable under the hybrid SQ(2) scheme
if and only if $\lambda \in \Lambda$.

\end{proposition} 

\begin{proof}
Let us assume that~\eqn{eq:cond_scheme1} holds. Now let
$p_i= \frac{\gamma_i C_i}{\sum_{j=1}^{M} \gamma_j C_j}$,
for all $i \in \cal{J}$. Using these values of $p_i$, $i \in \cal{J}$,
we have $\rho_i= \frac{\lambda}{\mu \sum_{j=1}^{M} \gamma_j C_j} < 1$.
%
%
Hence, condition~\eqn{eq:cond_scheme1} is sufficient.

Now let $\frac{\lambda}{\mu \sum_{j=1}^{M} \gamma_j C_j} \geq 1$.
For stability we must have $\rho_i  < 1$ for all $i \in \cal{J}$.
Hence, $\frac{\lambda}{\mu \sum_{j=1}^{M} \rho_j \gamma_j C_j} > 1$
which contradicts the fact that $\sum_{j=1}^{M} p_j=1$ or 
$\frac{\lambda}{\mu \sum_{j=1}^{M} \rho_j \gamma_j C_j} = 1$. 
Hence, condition~\eqn{eq:cond_scheme1}
is necessary.
\end{proof}

\begin{remark}
We have seen
that with the hybrid SQ(2) scheme it is possible 
to recover the stability region as defined in~\eqn{eq:cond_scheme1}. 
The intuition behind the loss of stability region
under the SQ(2) scheme is related to the fact that under uniform sampling, depending on the proportions
of fast and slow servers, one could frequently choose
slower servers even when they are heavily loaded and there are faster 
servers available with less congestion. 
Clearly, a biased 
sampling of the servers is one way to avoid this. 
The hybrid SQ(2) scheme provides the optimal way of choosing the bias.
\end{remark}

Henceforth we will assume that~\eqn{eq:cond_scheme1}
holds. We proceed to find the vector $\mf{p}^*=\cbrac{p_j^*, j\in \cal{J}}$
or equivalently the vector $\bs{\rho}^*=\cbrac{\rho_j^*, j \in \cal{J}}$
that minimizes the mean sojourn time of jobs in the limiting system
under the hybrid SQ(2) scheme. Similar to the SQ(2) scheme, it can be shown that
the mean sojourn time of jobs in the limiting system
under the hybrid SQ(2) scheme is given by 
$\bar{T}=\frac{1}{\lambda}\sum_{j=1}^{M} \sum_{k=1}^{\infty} \gamma_j P_k^{(j)}$,
where $P_k^{(j)}$, $j \in \cal{J}$ and $k \in \mb{Z}_{+}$, denotes
the stationary probability that a server with capacity $C_j$ in the
limiting system has atleast $k$ unfinished jobs under the hybrid SQ(2) scheme.
From the results of~\cite{Mitzenmacher_thesis, Vvedenskaya_inftran_1996, Bramson_randomized_load_balancing} 
it is known that for each $j \in \cal{J}$ and $k \in \mb{Z}_{+}$ we have $\mv{P}{k}{j}= \rho_j^{2^k-1}$.

Therefore, the overall mean sojourn time of jobs is given by $\bar{T}(\bs{\rho})=
\frac{1}{\lambda}\sum_{j=1}^{M} \gamma_j\sum_{k=1}^{\infty} \rho_j^{2^k-1}$.
We now formulate the mean sojourn time minimization problem
in terms of the loads $\rho_j$, $j\in \cal{J}$, as follows:

\begin{equation}
 \begin{aligned}
 & \underset{\bs{\rho}}{\text{Minimize}} & & \frac{1}{\lambda}\sum_{j \in \cal{J}} \gamma_j \sum_{k=1}^{\infty} \rho_j^{2^k-1} \\
 & \text{subject to} & & 0 \leq \rho_j < 1, \text{for all } j \in \cal{J}\\
 &&&  \sum_{j \in \cal{J}} \gamma_j C_j \rho_j = \frac{\lambda}{\mu}.
 \end{aligned}
 \label{opt:actual_osq}
\end{equation}
To characterize the solution of the convex problem 
defined in~\eqn{opt:actual_osq},
we assume without loss of generality that the
server capacities are ordered as follows:

\begin{equation}
C_1 \geq C_2 \geq \ldots \geq C_M
\label{eq:ordering_osq}
\end{equation}
Further, let $\mathcal{J}_{opt} \subseteq \cal{J}$ denote the 
index set of server capacities being
used in the optimal scheme.

\begin{proposition}
\label{thm:solution_osq}
Let $\Phi: \mb{R}_+ \rightarrow [0,1)$
be the inverse of the monotone mapping $\Phi^{-1}: [0,1)\rightarrow \mb{R}_+$
defined as $\Phi^{-1}(\rho)=\sum_{k=1}^{\infty} \brac{2^k-1} \rho^{2^k-2} < \sum_{x=1}^{\infty} x \rho^{x-1} <\infty$
for $0< \rho < 1$.
%
%
Further, for each $j \in \cal{J}$, let 
$\Psi_j: \mb{R}_{+} \rightarrow \mb{R}_{+}$
denote the inverse of the monotone mapping $\Psi_j^{-1}: \mb{R}_{+} \rightarrow \mb{R}_{+}$
defined as $\Psi_j^{-1}(\theta)=\mu \sum_{i=1}^{j} \gamma_i C_i \Phi(\theta C_i)$.
%
%
The index set of server capacities 
used in the hybrid SQ(2) 
scheme is then given by 
$\mathcal{J}_{\text{opt}}=\cbrac{1,2,\ldots,j^*}$,
where $j^*$ is given by

\begin{equation}
j^*=\sup\cbrac{j \in \mathcal{J}: \frac{1}{{C_j}} < \Psi_j(\lambda)}.
\label{eq:optset_osq}
\end{equation}
Moreover, the optimal traffic intensities $\rho_i^*$, for $i \in \cal{J}$ satisfy

\begin{equation}
\rho_i^*=\begin{cases} \Phi(\Psi_{j*}(\lambda)C_i), &\mbox{if } i \in \cal{J}_{\text{opt}}\\
           0, &\mbox{otherwise.}\end{cases}
\label{eq:optrho_osq}
\end{equation}
\end{proposition}

\begin{proof}
The Lagrangian associated with problem~\eqn{opt:actual_osq} is given by

\begin{equation}
\label{eq:lagrangian_osq}
L(\bs{\rho},\bs{\nu},\bs{\zeta},\theta)=\sum_{j=1}^{M} \gamma_j \sum_{k=1}^{\infty} \rho_i^{2^k-1} + \sum_{j=1}^{M} \nu_j (0-\rho_j)
+ \sum_{j=1}^{M} \zeta_j \brac{\rho_j-1}
+\theta\brac{\sum_{j=1}^{M} \gamma_j C_j \rho_j - \frac{\lambda}{\mu}},
\end{equation}
where $\bs{\nu} \geq \mf{0}$, $\bs{\zeta} \geq \mf{0}$, and $\theta \in \mb{R}$.
Since problem~\eqn{opt:actual_osq} is strictly convex and a feasible solution
exits (due to condition~\eqn{eq:cond_scheme1}), by Slater's condition, strong
duality is satisfied. Hence, the primal optimal solution $\bs{\rho}^*$ and the dual
optimal solution $\brac{\bs{\nu}^*, \bs{\zeta}^*, \theta^*}$ 
have zero duality gap if they
satisfy the KKT conditions given as follows:

\begin{align}
&\mf{0} \leq \bs{\rho}^* < \mf{1} \nonumber\\
&\sum_{j=1}^{M} \gamma_j C_j \rho_j^*=\frac{\lambda}{\mu} \nonumber\\
&\theta^* \in \mb{R}, \bs{\nu}^* \geq \mf{0}, \bs{\zeta}^* \geq \mf{0} \nonumber\\
\label{eq:comp_slack}
&\nu_j^* \rho_j^*=0, \zeta_j^*(\rho_j^*-1)=0 \text{ } \forall j \in \mathcal{J}\\
\label{eq:lag_grad}
\gamma_j \sum_{k=1}^{\infty} \brac{2^k-1} &\brac{\rho_j^*}^{2^k-2}
-\theta^* \gamma_j C_j - \nu_j^* + \zeta_j^*=0 \text{ } \forall j \in \mathcal{J} 
\end{align}
Since the objective function tends to infinity as $\rho_j^* \rightarrow 1$, 
for each $j \in \cal{J}$, it follows that necessarily $\bs{\rho}^* < \mf{1}$.
Hence, from~\eqn{eq:comp_slack}, $\bs{\zeta}^*=\mf{0}$. Since $\bs{\nu}^* \geq \mf{0}$,

\begin{equation}
\theta^* \leq \frac{1}{C_j} \sum_{k=1}^{\infty} \brac{2^k-1} (\rho_j^*)^{2^k-2}
\text{  }\forall j \in \cal{J}
\end{equation}
Further, by eliminating $\nu_j^*$ from~\eqn{eq:lag_grad} we obtain

\begin{equation}
\brac{\sum_{k=1}^{\infty} \brac{2^k-1} (\rho_j^*)^{2^k-2}-\theta^* C_j} \rho_j^*=0
\label{eq:inter}
\end{equation}
Thus, if, for some $j \in \cal{J}$, $\theta^* > \frac{1}{C_j}$, then
$\rho_j^* > 0$. Therefore, from~\eqn{eq:inter} and from the definition
of the map $\Phi$ we have $\rho_j^*=\Phi(\theta^*C_j)$.
%
%
If $\theta^* \leq \frac{1}{C_j}$ for some $j \in \cal{J}$, then
$\rho_j^*=0$. Hence, we have

\begin{equation}
\rho_j^*=\begin{cases} \Phi(\theta^*C_j), &\mbox{if } \frac{1}{C_j} < \theta^*\\
           0, &\mbox{otherwise.}\end{cases}
\label{eq:optrho_osq1}
\end{equation}
To find $\theta^*$, we use the equality constraint in~\eqn{opt:actual_osq}.
If the first $j^*$ server capacities are used in the optimal
SQ(2) scheme then 

\begin{equation}
\sum_{j=1}^{j^*} \gamma_j C_j \Phi(\theta^* C_j)= \frac{\lambda}{\mu}
\end{equation} 
Hence by definition of the map $\Psi_j$,

\begin{equation}
\theta^*= \Psi_{j*}(\lambda),
\end{equation}
where $j^*$ is defined as in~\eqn{eq:optset_osq}.
\end{proof}

The optimal routing probabilities $p_j^*$, $j \in \cal{J}$, and the
minimum mean sojourn time $\bar{T}^*$ can be found from
Proposition~\ref{thm:solution_osq} by using the relations $\rho_j^*=\frac{p_j^* \lambda}{\gamma_j \mu C_j}$ and $\bar{T}^*= \frac{1}{\lambda}\sum_{i=1}^{j*} \gamma_i \sum_{k=1}^{\infty} (\rho_i^*)^{2^k-1}$, respectively.

\begin{remark}
One drawback of the hybrid SQ(2) scheme 
is that the arrival rates need to be estimated
to obtain the optimal sampling biases that would minimize
the average delay. However, if one is only interested in maximize 
the stability region, then a much simpler biasing scheme exists in which the
knowledge of the server speeds and their proportions is sufficient. 
Indeed, it is easy to see that choosing the sampling probabilities as:
$p_i=\frac{\gamma_i C_i}{\sum_{j=1}^{M} \gamma_j C_j}$ for each $i \in \cal{J}$
gives $\Lambda$ as the stability region.

Such a sampling bias will not necessarily minimize the average delay.
\end{remark}

\section{Numerical results}
\label{sec:numerics}

In this section, we present simulation results
to compare the different load balancing schemes
considered in this paper. The results also 
indicate the accuracy of the asymptotic analyses
of the SQ(2) and the hybrid SQ(2) schemes in
predicting their performance in a finite system
of servers. We set $\mu=1$ in all our simulations. We also plot the simulation results 
for the SQ(5) scheme whose analysis and characterization is extremely complicated in 
the heterogeneous case but 
can be shown to be superior to the SQ(2) case by coupling arguments. But as argued by \cite{Mitzenmacher_IEEE_2001,Vvedenskaya_inftran_1996} 
most of the gains are achieved 
(super-exponential decay of the tail distributions) by considering 
the SQ(2) scheme - the focus of this paper.

\begin{figure}
 \centering
 \includegraphics[height=0.4\columnwidth,
  keepaspectratio]{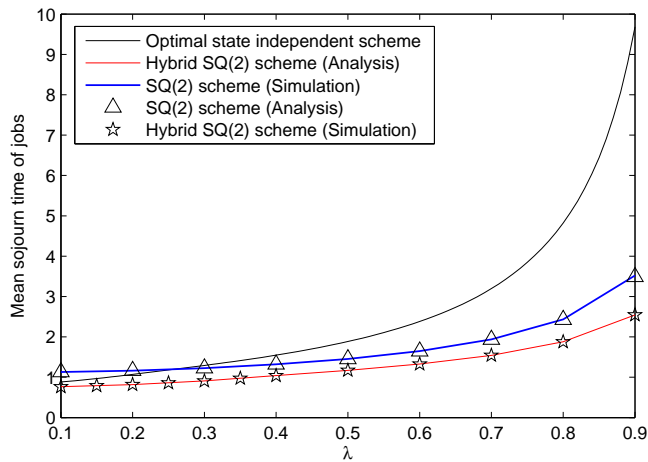}
 \caption{Mean sojourn time jobs as a function of $\lambda$ for $C_1=2/3$, $C_2=4/3$, $N=200$ and $\gamma_1=\gamma_2=1/2$}
 \label{fig:del_lam_het}
 \end{figure}
 
\begin{figure}
 \centering
 \includegraphics[height=0.4\columnwidth,
  keepaspectratio]{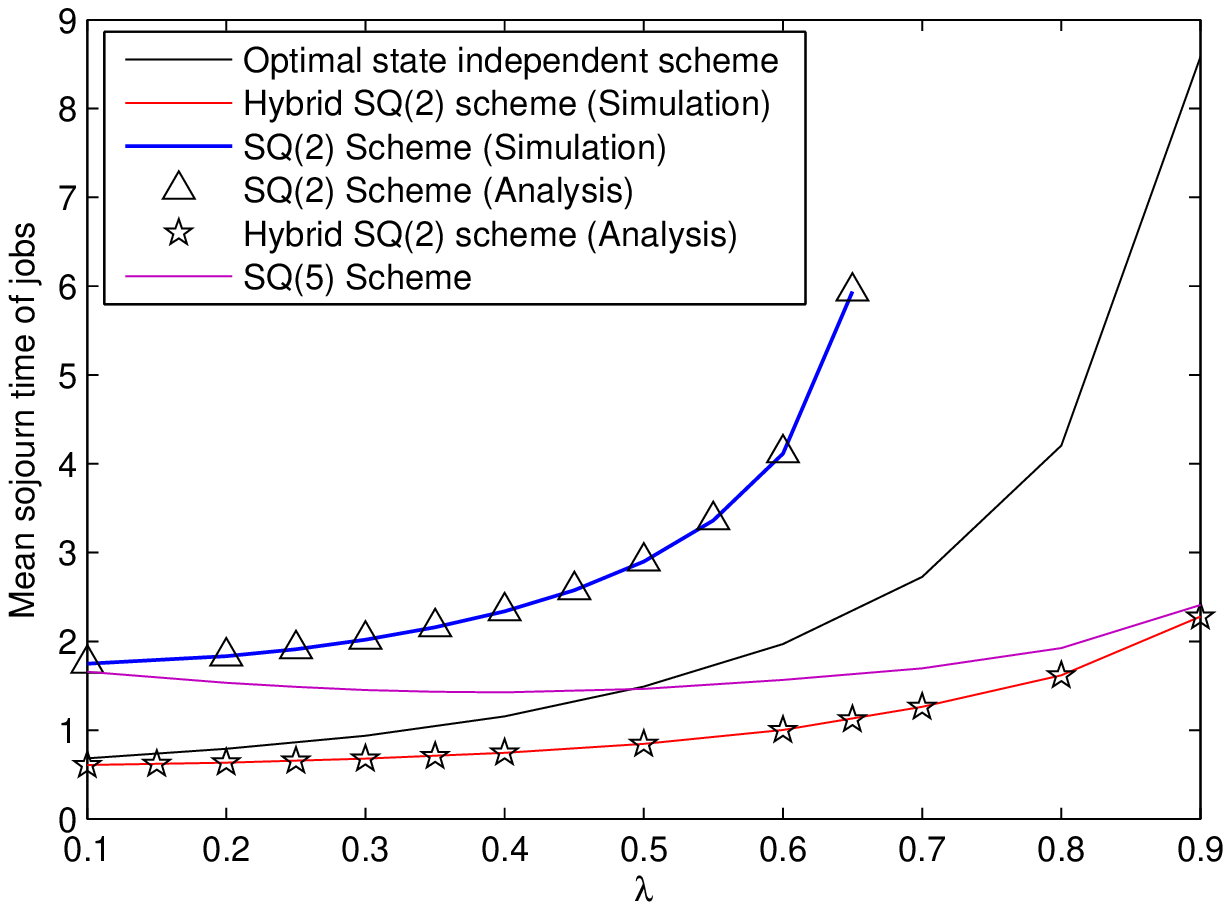}
 \caption{Mean sojourn time jobs as a function of $\lambda$ for $C_1=1/3$, $C_2=5/3$, $N=200$, and $\gamma_1=\gamma_2=1/2$}
 \label{fig:del_lam_het1}
 \end{figure}  
 
\begin{figure}
 \centering
 \includegraphics[height=0.4\columnwidth,
  keepaspectratio]{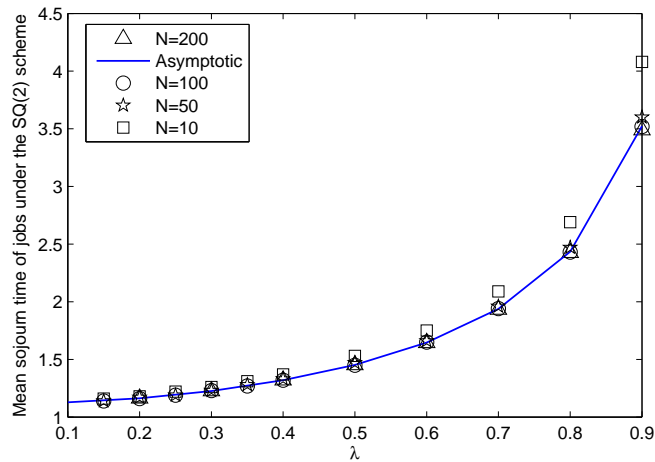}
 \caption{Mean sojourn time jobs under the SQ(2) scheme as a function of $\lambda$ for different values of $N$}
 \label{fig:convergence}
 \end{figure}

We first set 
$C_1=4/3$, $C_2=2/3$, $N=200$
and $\gamma_1=\gamma_2=\frac{1}{2}$. 
Using conditions~\eqn{eq:cond_scheme1},~\eqn{eq:cond_scheme2_2} 
it is found that $\Lambda=\Lambda_{\infty}=\cbrac{0 \leq \lambda < 1}$.
In Figure~\ref{fig:del_lam_het}, we plot the mean sojourn time
jobs in the system as a function of the normalized arrival rate, $\lambda$,
for the three schemes. 
It is observed from the plot that 
the SQ(2) scheme performs better than Scheme~1 for higher values of $\lambda$
and the hybrid SQ(2) scheme results in the least mean sojourn time of jobs among all the three schemes.

The performance of the SQ(2) scheme
may not always be better than that of Scheme~1. To demonstrate this fact
we choose a second set of parameter values as follows:
$C_1=5/3$, $C_2=1/3$, $N=200$, and $\gamma_1=\gamma_2=1/2$. 
Under this parameter setting, we have
$\Lambda=\cbrac{0 \leq \lambda <1}$ and $\Lambda_{\infty}=\cbrac{0 \leq \lambda < 2/3}$. 
Therefore, in this setting, the asymptotic stability region under the SQ(2) scheme is
a strict subset of the stability region under Scheme~1 and the hybrid SQ(2) scheme.
In Figure~\ref{fig:del_lam_het1}, we plot the average response
time of jobs as a function of $\lambda$
for the three schemes and the SQ(5) scheme.
We see that the mean response time of jobs 
is lower in Scheme~1 than that in
the SQ(2) scheme. As in the previous setting, the hybrid SQ(2) scheme
outperforms Scheme~1 and the SQ(2) scheme. 
Furthermore, the SQ(5) scheme outperforms the SQ(2) scheme.

In Figure~\ref{fig:convergence}, we plot
the mean sojourn time of jobs as a function of $\lambda$
for different values of the system size $N$. The plots
are obtained for the first parameter setting where $\Lambda=\Lambda_{\infty}$.
We observe a good match between the asymptotic analysis and the 
simulation results for $N=50,100, 200$. 
The simulation results deviate from the analysis
for $N=10$ where the percentage of deviation is between 5-15\%. 
This leads us to believe that the mean-field results derived in this paper 
can be used to accurately predict the behavior of the
schemes even for moderate number of servers.

The asymptotic insensitivity of 
the SQ(2) scheme.  
is numerically validated in Table~\ref{N_100}, where the the mean sojourn
time of jobs were obtained 
for the parameter setting $C_1=4/3$, $C_2=2/3$, $N=200$
and $\gamma_1=\gamma_2=\frac{1}{2}$.
We chose the following two distributions: 
i) constant, with distribution satisfying $F(x)=0$ for $0\leq x < 1$,
and $F(x)=1$, otherwise.
ii) power law, with distribution satisfying 
$F(x)=1-1/4x^2$ for $x \geq \frac{1}{2}$ and $F(x)=0$, otherwise.
It is seen that there is insignificant change in the mean sojourn
time of jobs when
the job length distribution type is changed.  
The results, therefore, justify the asymptotic independence assumption
stated in~\ref{sec:scheme2}.

\begin{table}[ht]
 \renewcommand{\arraystretch}{1.5}
 \caption{Insensitivity of the SQ(2) Scheme}
 \centering
 \begin{tabular}{c c c c}
   \hline
   \bf{$\lambda$} & \begin{tabular}[c]{@{}c@{}} Mean sojourn time $\bar{T}$\\Theoretical \end{tabular}  & \begin{tabular}[l]{@{}c@{}}Constant\\Simulation\end{tabular} & \begin{tabular}[c]{@{}c@{}}Power Law\\Simulation\end{tabular}\\
   \hline
   0.2 & 1.1614 & 1.1623 & 1.1620\\
   \hline
   0.3 & 1.2257 & 1.2257 & 1.2261\\
   \hline
   0.5 & 1.4547 & 1.4533 & 1.4550\\
   \hline
   0.7 & 1.9375 & 1.9377 & 1.9380\\
   \hline
   0.8 & 2.4265 & 2.4335 & 2.4330\\
   \hline
   0.9 & 3.5300 & 3.5204 & 3.5210\\
   \hline
 \end{tabular}
 \label{N_100}
 \end{table}
 

\section{Conclusion}
\label{sec:conclusions}

In this paper, we considered randomized load balancing
schemes for large heterogeneous processor sharing systems. 
It was shown that, as in the homogeneous case,
the asymptotic stationary tail distribution of loads at each server
decreases doubly exponentially and is insensitive to the type
of job length distribution under the SQ(2) scheme. 
However, unlike the homogeneous case,
in the heterogeneous case, the SQ(2) scheme has a smaller stability region 
than the average capacity of the system.
We have shown that the maximal stability region can be fully recovered
by using a scheme that combines the SQ(2) scheme with the state independent 
scheme and provides the best mean sojourn time 
behaviour among all the schemes considered in the paper.

\appendix

\subsection{Proof of Proposition~\ref{thm:stability}}
\label{proof:thm_stability}

From condition (1.2) of~\cite{Bramson_stability_AAP}  for any finite
value of $N$, the system is stable under the SQ(2) scheme
if the following condition is satisfied:

\begin{equation}
\max_{\cal{B} \subseteq \cal{S}} \cbrac{ {\brac{\sum_{i \in \cal{B}} C_{(i)}}}^{-1} \frac{N \lambda}{\mu} \frac{\binom{\vert \cal{B} \vert}{2}}{\binom{N}{2}}} < 1,
\label{eq:cond1}
\end{equation} 
where $\mathcal{S}=\cbrac{1,2,\ldots,N}$ denotes the index set of servers, $\cal{B} \subseteq S$
is a subset of servers of size at least $2$, 
and $C_{(k)} \in \cal{C}$ denotes the capacity of the $k^{\textrm{th}}$
server in the system. Thus, for $N=kN^*$, the set $\Lambda_k$ is given by

\begin{equation}
\Lambda_k=\cbrac{\lambda > 0 : {\brac{\sum_{i \in \cal{B}} C_{(i)}}}^{-1} \frac{N \lambda}{\mu} \frac{\binom{\vert \cal{B} \vert}{2}}{\binom{N}{2}} < 1 \text{ }\forall \text{ }\mathcal{B} \subseteq \mathcal{S}_k}
\label{eq:set_lam}
\end{equation}
where $\mathcal{S}_k=\cbrac{1,2,\ldots,kN^*}$. Clearly, for integers $l$ and $k$,
with $l \geq k$, we have $\mathcal{S}_k \subseteq \mathcal{S}_l$. Hence, if
$\mathcal{B} \subseteq \mathcal{S}_k$, then $\mathcal{B} \subseteq \mathcal{S}_l$.
Therefore, from~\eqn{eq:set_lam} it is clear that $\lambda \in \Lambda_l$ implies
$\lambda \in \Lambda_k$. Consequently, for $l \geq k$ we have $\Lambda_k \supseteq \Lambda_l$. 
Further, if we set $\mathcal{B}=\mathcal{S}$ in~\eqn{eq:cond1} then we get~\eqn{eq:cond_scheme1}.
Hence, for all $k \in \mb{N}$, $\Lambda \supseteq \Lambda_k$. This proves~\eqn{eq:set_seq}.

To prove~\eqn{eq:lam_infty}, let us consider a finite value of $N$
and a set $\cal{I} \subseteq \cal{J}$. Let $\mathcal{B}_{\mathcal{I}} \subseteq \mathcal{S}$ 
be a subset of servers in which there are $a_i$ ($0 < a_i \leq N\gamma_i$) servers of capacity $C_i$ for each $i \in \cal{I}$.
It can be easily checked that $\frac{\brac{\sum_{i \in \cal{I}} a_i}\brac{\sum_{i \in \cal{I}} a_i-1}}
{\sum_{i\in\cal{I}} a_i C_i}$ is an increasing function in each of the variables $a_i$.
Therefore, we have

\begin{align*}
{\brac{\sum_{i \in \mathcal{B}_\mathcal{I}} C_{(i)}}}^{-1} \frac{N \lambda}{\mu} \frac{\binom{\vert \mathcal{B}_{\mathcal{I}} \vert}{2}}{\binom{N}{2}}
& = \frac{\lambda}{\mu}\frac{\brac{\sum_{i \in \cal{I}} a_i}\brac{\sum_{i \in \cal{I}} a_i-1}}
{\brac{\sum_{i\in\cal{I}} a_i C_i}\brac{N-1}}\\
& \leq \frac{\lambda}{\mu}\frac{\brac{\sum_{i \in \cal{I}}N \gamma_i}\brac{\sum_{i \in \cal{I}} N \gamma_i-1}}
{\brac{\sum_{i\in\cal{I}} N \gamma_i C_i}\brac{N-1}}\\
& \leq \frac{\lambda}{\mu}\frac{\brac{\sum_{i \in \cal{I}}N \gamma_i}\brac{\sum_{i \in \cal{I}} N \gamma_i}}
{\brac{\sum_{i\in\cal{I}} N \gamma_i C_i}\brac{N}}\\
& = \frac{\lambda}{\mu}\frac{\brac{\sum_{i \in \cal{I}} \gamma_i}^2}
{\brac{\sum_{i\in\cal{I}} \gamma_i C_i}}
\end{align*} 
The first equality follows from simplifying the expression on the L.H.S.
The second inequality follows from the first since we have
$\frac{N\alpha-1}{N-1}\leq \frac{N\alpha}{N}=\alpha$ 
for $\alpha \leq 1$.
Hence, $\lambda \in \Lambda_{\infty}$ implies
${\brac{\sum_{i \in \mathcal{B}_\mathcal{I}} C_{(i)}}}^{-1} \frac{N \lambda}{\mu} \frac{\binom{\vert \mathcal{B}_{\mathcal{I}} \vert}{2}}{\binom{N}{2}} < 1$.
As this is true for any $\cal{I} \subseteq \cal{J}$ and any $N$,
we have that $\Lambda_{\infty} \subseteq \Lambda_k$ for all $k \in \mb{N}$.
Hence, $\Lambda_{\infty} \subseteq \cap_{k=1}^{\infty} \Lambda_k$. To prove the
reverse inclusion, consider $\lambda \in \cap_{k=1}^{\infty} \Lambda_k$. 
For $\cal{I} \subseteq \cal{J}$, consider a set $\mathcal{B}_{\mathcal{I}}^{(N)}$
which contains all the $N \gamma_i$ servers of capacity $C_i$ for each $i \in \cal{I}$. 
Since $\lambda \in \Lambda_k$ for all $k \in \mb{N}$, we have $\lim_{N \rightarrow \infty}{\brac{\sum_{i \in \mathcal{B}_\mathcal{I}} C_{(i)}}}^{-1} \frac{N \lambda}{\mu} \frac{\binom{\vert \mathcal{B}_{\mathcal{I}}^{(N)} \vert}{2}}{\binom{N}{2}}<1$, which
is equivalent to the condition $\frac{\lambda}{\mu}\frac{\brac{\sum_{i \in \cal{I}} \gamma_i}^2}
{\brac{\sum_{i\in\cal{I}} \gamma_i C_i}} < 1$. As this is true for all $\cal{I} \subseteq \cal{J}$,
we have $\lambda \in \Lambda_{\infty}$. Hence, $\Lambda_{\infty}=\cap_{k=1}^{\infty} \Lambda_k$ as required.

\subsection{Proof of Proposition~\ref{thm:tail_het_properties}}
\label{proof:tail_het_properties}
i) Let $\mf{P}$ satisfy the hypothesis of the proposition. Hence, from~\eqn{eq:diff4},
we have that, for each $l \in \mb{Z}_+$ and $j \in \cal{J}$,
\begin{equation}
\mv{P}{l+1}{j}-\mv{P}{l+2}{j}=\nu_j \brac{\mv{P}{l}{j}-\mv{P}{l+1}{j}}
\sum_{i=1}^{M} \gamma_i \brac{\mv{P}{l}{i}+\mv{P}{l+1}{i}}.
\label{eq:rec}
\end{equation}
Since by hypothesis $\mv{P}{l}{j} \rightarrow 0$ as $l \rightarrow \infty$,
adding the above equations for $l \geq k$ yields~\eqn{eq:tail_het} upon simplification.

ii) Equation~\eqn{eq:nice} is a direct consequence of~\eqn{eq:tail_het}.

iii) From~\eqn{eq:nice} we obtain
$\frac{\gamma_j\mv{P}{k+1}{j}}{\nu_j} \leq \brac{\sum_{j=1}^{M}\gamma_j \mv{P}{k}{j}}^2
\leq \brac{\tilde{P}_k}^2$,
%
where $\tilde{P}_k=\max_{1 \leq j \leq M} \mv{P}{k}{j}$. Thus,
we have $\mv{P}{k+1}{j} \leq \delta \tilde{P}_k$, where $\delta=\tilde{P}_k \max_{1 \leq j \leq M} (\nu_j/\gamma_j)$. 
Since by hypothesis, for each $j$, $\mv{P}{k}{j} \rightarrow 0$ 
as $k \rightarrow \infty$, one can choose $k$ sufficiently large such that $\delta < 1$. Hence, we have $\brac{\max_{1 \leq j \leq M} \mv{P}{k+1}{j}} \leq \delta \tilde{P}_k$.
Similarly we have, $\brac{\max_{1 \leq j \leq M} \mv{P}{k+n}{j}} \leq \delta^{2^n-1} \tilde{P}_k$. This proves that the sequence $\cbrac{\mv{P}{k}{j}, k \in \mb{Z}_+}$
decreases doubly exponentially for each $j$.

\subsection{Proof of Proposition~\ref{thm:diff_eqn}}
\label{proof:diff_eqn}

i) Define $\theta(x)=[\min(x,1)]_{+}$, where $[z]_{+}=\max(0,z)$.
Now, we consider the following modification of~\eqn{eq:diff1}-\eqn{eq:diff4}.

\begin{align}
\mathbf{u}(0) &= \mathbf{g},\label{eq:diff1'}\\
\dot{\mathbf{u}}(t) &= \tilde{\mathbf{h}}(\mathbf{u}(t)),
\label{eq:diff2'}
\end{align} 
where for $1 \leq j \leq M$,

\begin{align}
\tilde{h}^{(j)}_0(\mathbf{u}) &= 0,\label{eq:diff3'}\\
\tilde{h}^{(j)}_n(\mathbf{u}) &= \lambda \sbrac{\theta(u^{(j)}_{n-1})-\theta(u^{(j)}_n)}_{+} \sum_{i=1}^{M} \gamma_i \sbrac{\theta(u^{(i)}_{n-1})+\theta(u^{(i)}_n)}-\mu C_j \sbrac{\theta(u^{(j)}_n)-\theta(u^{(j)}_{n+1})}_{+}
\label{eq:diff4'}
\end{align}
for all $n \geq 1$. Note that the right hand side of~\eqn{eq:diff4} and~\eqn{eq:diff4'}
are equal if $\mathbf{u} \in \Ub^M$. Therefore, the two systems have the same
solution in $\Ub^M$. Also if $\mf{g} \in \Ub^M$, then any solution of
the modified system remains within $\Ub^M$.
This is because of the facts that if $u^{(j)}_n(t)=u^{(j)}_{n+1}(t)$ for
some $j$, $n$, $t$, then $h^{(j)}_n(\mf{u}(t)) \geq 0$ and 
$h^{(j)}_{n+1}(\mf{u}(t)) \leq 0$, and if $\mv{u}{n}{j}(t)=0$
for some $j$, $n$, $t$, then $\mv{h}{n}{j}(\mf{u}) \geq 0$.
Hence, to prove the uniqueness of solution of
\eqn{eq:diff1}-\eqn{eq:diff4}, we need to show
that the modified system~\eqn{eq:diff1'}-\eqn{eq:diff4'}
has a unique solution in $(\mb{R}^{\mb{Z}_{+}})^{M}$.

Using the norm defined in~\eqn{eq:norm} and the facts
that $\abs{x_+-y_+} \leq \abs{x-y}$ for any $x,y \in \mb{R}$,
$\abs{a_1b_1-a_2b_2} \leq \abs{a_1-a_2}+\abs{b_1-b_2}$ for any 
$a_1,a_2,b_1,b_2 \in [0,1]$, and $\abs{\theta(x)-\theta(y)} \leq \abs{x-y}$
for any $x,y \in \mb{R}$ we obtain

\begin{align}
\norm{\mf{\tilde{h}}(\mf{u})} &\leq K_1 \label{eq:bound}\\
\norm{\mf{\tilde{h}}(\mf{u}_1)-\mf{\tilde{h}}(\mf{u}_2)} &\leq K_2 \norm{\mf{u_1}-\mf{u_2}},
\label{eq:lip}
\end{align}
where $K_1$ and $K_2$ are constants defined as
$K_1=2\lambda+\mu (\max_{1 \leq j \leq M} C_j)$
and $K_2= 8\lambda+ 2 \mu (\max_{1 \leq j \leq M} C_j)$.
The uniqueness follows from inequalities~\eqn{eq:bound}
and~\eqn{eq:lip} by using Picard's successive approximation
technique since $\Ub^M$ is complete under the norm defined in~\eqn{eq:norm}.

ii) For ease of exposition we provide a proof for the $M=2$ case. 
The proof can be extended to any $M \geq 2$.

We note that if there exists $\mf{P} \in \Ub^M$ such that the sequences $\cbrac{\mv{P}{l}{1}, l\in \mb{Z}_+}$
and $\cbrac{\mv{P}{l}{2}, l\in \mb{Z}_+}$ satisfy the recursive relation~\eqn{eq:rec}
for all $l \in \mb{Z}_+$, $j=1,2$, then 
it must be an equilibrium point of the system~\eqn{eq:diff1}-\eqn{eq:diff4}.
Moreover, if $\mv{P}{l}{1}, \mv{P}{l}{2} \downarrow 0$
as $l \rightarrow \infty$, then by Proposition~\ref{thm:tail_het_properties}.iii), such $\mf{P}$ must also lie in the space $\U^M$.
We now proceed to prove that such $\mf{P}$ exists.

We construct the sequences $\cbrac{\mv{P}{l}{1}(\alpha), l\in \mb{Z}_+}$
and $\cbrac{\mv{P}{l}{2}(\alpha), l\in \mb{Z}_+}$ as functions of the
real variable $\alpha$ as follows: $\mv{P}{0}{1}(\alpha)=\mv{P}{0}{2}(\alpha)=1$,
$\mv{P}{1}{1}(\alpha)=\alpha$, $\mv{P}{1}{2}(\alpha)=\frac{\nu_2}{\gamma_2}\brac{1-\frac{\gamma_1}{\nu_1}\alpha}$,
and for $l \geq 0$ and $j=1,2$ the following recursive relationship holds

\begin{equation}
\mv{P}{l+2}{j}(\alpha)=\mv{P}{l+1}{j}(\alpha)-\nu_j\brac{\mv{P}{l}{j}(\alpha)-\mv{P}{l+1}{j}(\alpha)} \brac{\sum_{i=1}^{2}\gamma_i\brac{\mv{P}{l}{i}(\alpha)+\mv{P}{l+1}{i}(\alpha)}}.
\label{eq:rec2}
\end{equation}
Note that the above relation is same as~\eqref{eq:rec}. We show that 
there exists some value of $\alpha$, such that both $\cbrac{\mv{P}{l}{1}(\alpha), l\in \mb{Z}_+}$
and $\cbrac{\mv{P}{l}{2}(\alpha), l\in \mb{Z}_+}$ are non-negative, decreasing sequences
(monotonicity will follow from non-negativity by virtue of~\eqref{eq:rec2}). 
It can be shown
from the relations $\frac{\gamma_1}{\nu_1}\mv{P}{1}{1}(\alpha)+\frac{\gamma_2}{\nu_2}\mv{P}{1}{2}(\alpha)=1$
and~\eqn{eq:rec2} that

\begin{equation}
\sum_{j=1}^{2} \frac{\gamma_j}{\nu_j} \mv{P}{l+1}{j}(\alpha)= \brac{\sum_{j=1}^{2} \gamma_j \mv{P}{l}{j}(\alpha)}^2 \text{ for } l \geq 0
\label{eq:nice2}
\end{equation}
From condition~\eqref{eq:cond_scheme2_2} and the construction above we see
that for $\alpha \in \brac{\max\brac{0,\frac{\nu_1}{\gamma_1}\brac{1-\frac{\gamma_2}{\nu_2}}},
\min\brac{1,\frac{\nu_1}{\gamma_1}}}$ we have $1=\mv{P}{0}{j}(\alpha) > \mv{P}{1}{j}(\alpha) >0$
for $j=1,2$. By using~\eqref{eq:rec2} for $l=2$ and $j=1$, we have that $\mv{P}{2}{1}(\alpha) < 0$
for $\alpha=0$ and $\mv{P}{2}{1}(\alpha) > 0$ for $\alpha=1,\frac{\nu_1}{\gamma_1}$.
Hence there must exist at least one root of $\mv{P}{2}{1}(\alpha)$ in $\brac{0, \min \brac{1,\frac{\nu_1}{\gamma_1}}}$.
Let the maximum of these roots be $r_1$. Therefore, if $\alpha \in \brac{\max\brac{r_1,\frac{\nu_1}{\gamma_1}\brac{1-\frac{\gamma_2}{\nu_2}}},\min\brac{1,\frac{\nu_1}{\gamma_1}}}$ then $1=\mv{P}{0}{1}(\alpha) > \mv{P}{1}{1}(\alpha) > \mv{P}{2}{1}(\alpha) >0$.
Similarly, for $\alpha=r_1,\frac{\nu_1}{\gamma_1}\brac{1-\frac{\gamma_2}{\nu_2}}$, we have $\mv{P}{2}{2}(\alpha) > 0$ and for $\alpha=\frac{\nu_1}{\gamma_1}$
we have $\mv{P}{2}{2}(\alpha) < 0$. Therefore, there must exist a root 
of $\mv{P}{2}{2}(\alpha)$ in $\alpha \in \brac{\max\brac{r_1,\frac{\nu_1}{\gamma_1}\brac{1-\frac{\gamma_2}{\nu_2}}},\frac{\nu_1}{\gamma_1}}$. If we denote the minimum of these roots by $r_2$, 
then for $\alpha \in \brac{\max\brac{r_1,\frac{\nu_1}{\gamma_1}\brac{1-\frac{\gamma_2}{\nu_2}}},
\min\brac{r_2, 1}}$ we get $1=\mv{P}{0}{j}(\alpha) > \mv{P}{1}{j}(\alpha) > \mv{P}{2}{j}(\alpha) >0$ for $j=1,2$. Continuing in this way we can always get a range of $\alpha \in 
\brac{\max\brac{r_{2k+1},\frac{\nu_1}{\gamma_1}\brac{1-\frac{\gamma_2}{\nu_2}}},
\min\brac{r_{2k+2}, 1}}$ such that $1=\mv{P}{0}{j}(\alpha) > \mv{P}{1}{j}(\alpha) > \mv{P}{2}{j}(\alpha) > \ldots > \mv{P}{k+2}{j}(\alpha)>0$ for $j=1,2$. 
Hence, there exists a value of $\alpha$ for which the sequences $\cbrac{\mv{P}{l}{1}(\alpha), l\in \mb{Z}_+}$
and $\cbrac{\mv{P}{l}{2}(\alpha), l\in \mb{Z}_+}$ are non-negative, monotonically decreasing sequences
in $[0,1]$ starting at $1$ satisfying~\eqref{eq:rec}. In other words there exists $\alpha$
for which $\mf{P}(\alpha)=\cbrac{\mv{P}{l}{j}(\alpha), l\in\mb{Z}_+, j=1,2}$ is in $\Ub^M$
and is an equilibrium point of the system~\eqref{eq:diff1}-\eqref{eq:diff4}. We now prove that
for such $\mf{P}(\alpha)$, $\mv{P}{l}{j}(\alpha) \rightarrow 0$ as $l \rightarrow \infty$ for $j=1,2$. 

We have seen that there exists a value of $\alpha$ such that
the sequences $\cbrac{\mv{P}{l}{1}(\alpha), l\in \mb{Z}_+}$
and $\cbrac{\mv{P}{l}{2}(\alpha), l\in \mb{Z}_+}$
are non-negative and monotonically decreasing sequences
in $[0,1]$ starting at $1$. Hence, by monotone convergence theorem,
both these sequences must converge in $[0,1]$.
Let $\mv{P}{l}{1}(\alpha) \rightarrow \zeta_1 \in [0,1]$ and $\mv{P}{l}{2}(\alpha) \rightarrow \zeta_2\in [0,1]$
as $l \rightarrow \infty$. Hence, by taking limit as $l \rightarrow \infty$
on both sides of relation~\eqref{eq:nice2}, we obtain

\begin{equation}
\sum_{j=1}^{2} \frac{\gamma_j}{\nu_j} \zeta_j= \brac{\sum_{j=1}^{2} \gamma_j \zeta_j}^2 
\label{eq:lim}
\end{equation}
Expressing the above equation as a quadratic equation $q(\zeta_1)$ in $\zeta_1$
we see that that $q(0)=\gamma_1\zeta_2 \brac{\gamma_2\zeta_2-\frac{1}{\nu_2}} < 0$
for $0< \zeta_2 \leq 1$ since by~\eqref{eq:cond_scheme2_2} $\gamma_2 \nu_2 <1$. 
Further, $q(1)=\gamma_2^2 \zeta_2^2+\brac{2\gamma_1\gamma_2-\frac{\gamma_2}{\nu_2}}\zeta_2
+\brac{\gamma_1^2-\frac{\gamma_1}{\nu_1}}$. By using the stability condition~\eqref{eq:cond_scheme2_2}
it can be easily shown that $q(1) < 0$ if $0< \zeta_2 \leq 1$. Hence, either both roots 
or no roots of $q(\zeta_1)=0$ must lie in $[0,1]$. Now, since the product of the roots
of $q(\zeta_1)=0$ is $q(0)/\gamma_1^2 <0$ for $0 < \zeta_2 \leq 1$ we conclude
that there is no root of $q(\zeta_1)=0$ in $[0,1]$ if $0 < \zeta_2 \leq 1$.
Hence, $\zeta_2=0$. For $\zeta_2=0$, 
the only solution of $q(\zeta_1)=0$ in $[0,1]$ is $\zeta_1=0$.
Therefore, we conclude $\zeta_1=\zeta_2=0$. Therefore, there exists a value
of $\alpha$ such that $\mv{P}{l}{1}(\alpha), \mv{P}{l}{2}(\alpha) \downarrow 0$ as $l \rightarrow \infty$.
Thus, there exists $\alpha$ such that $\mf{P}(\alpha) \in \U^M$ and is an equilibrium point of~\eqn{eq:diff1}-\eqn{eq:diff4}.
The uniqueness will follow from part (iii) of the proposition due to uniqueness of the limit.

iii) The proof is similar to the proof of Theorem 1.(iii) of~\cite{Martin_AAP_1999} and hence
is omitted to conserve space.

\bibliographystyle{ieeetr}

\bibliography{load_balance}

\begin{IEEEbiographynophoto}
{\bf Arpan Mukhopadhyay} received his Bachelors of Engineering  (B.E) degree in Electronics and Telecommunication Engineering from Jadavpur University, Culcutta, India in 2009, and his Master of Engineering (M.E) degree in Telecommunication from Indian Institute of Science, Bangalore, India in 2011. 

He is currently pursuing his Ph.D in Electrical and Computer Engineering at the University of Waterloo, Canada. His current areas of research are broadly stochastic modeling and analysis of large distributed networks, queuing theory, and network resource allocation and optimization algorithms. 
\end{IEEEbiographynophoto}

\begin{IEEEbiographynophoto}
{\bf Ravi Mazumdar} (F'05) was born in April 1955 in Bangalore, India. He 
obtained the B.Tech. in Electrical Engineering from the Indian Institute of 
Technology, Bombay, India in 1977, the M.Sc. DIC in Control Systems from 
Imperial College, London, U.K. in 1978 and the Ph.D. in Systems Science from 
the University of California, Los Angeles, USA in 1983. 

He is currently a University Research Chair Professor of Electrical and Computer Engineering  at the University of Waterloo, Waterloo, Canada and an Adjunct Professor of ECE at Purdue University. He has served on the faculties of Columbia University (NY, USA) and INRS-Telecommunications (Montreal, Canada) . He held a Chair in Operational Research and Stochastic Systems in the Dept. of Mathematics at the University of Essex (Colchester, UK), and from 1999-2005 was Professor of ECE at Purdue University (West Lafayette, USA). He has held visiting positions and sabbatical leaves at UCLA, the University of Twente (Netherlands), the Indian Institute of Science (Bangalore), the Ecole Nationale Superieure des Telecommunications (Paris), INRIA (Paris) and the University of California, Berkeley.

He is a Fellow of the IEEE and the Royal Statistical Society. He is  a member of the working groups WG6.3 and 7.1 of the IFIP and a member of SIAM and the IMS. He is a recipient of  the IEEE INFOCOM 2006  Best Paper Award and was runner-up for the Best Paper at INFOCOM 1998, 

His research interests are in applied probability, stochastic analysis, optimization,  and game theory with applications to network science, traffic engineering, filtering theory, and wireless systems.
\end{IEEEbiographynophoto}

\end{document}